\documentclass{article}

\usepackage[preprint]{neurips_2020}

\usepackage[utf8]{inputenc} 
\usepackage[T1]{fontenc}    
\usepackage{hyperref}       

\usepackage{url}            
\usepackage{booktabs}       
\usepackage{amsfonts}       
\usepackage{nicefrac}       
\usepackage{microtype}      

\usepackage{amsmath}
\usepackage{amssymb}
\usepackage{amsthm}
\usepackage{amsfonts}
\usepackage[pdftex]{graphicx}

\usepackage[dvipsnames]{xcolor}
\definecolor{mydarkblue}{rgb}{0.03,0.2,0.4}

\definecolor{mypink}{rgb}{0.8392156862745098, 0.32941176470588235, 0.474509803921568}
\definecolor{myblue}{rgb}{0.29411764705882354, 0.4549019607843137, 0.6941176470588235}
\definecolor{mygreen}{rgb}{0.5372549019607843, 0.7607843137254902, 0.6470588235294118}
\definecolor{myorange}{rgb}{0.9058823529411765, 0.6549019607843137, 0.36470588235294116}

\hypersetup{
    colorlinks,
    linkcolor={black},
    citecolor={mydarkblue},
    urlcolor={mydarkblue}
}

\usepackage[disable]{todonotes}

\newcommand{\vv}[1]{\todo[color={orange!20}, inline]{Victor: #1}}

\newcommand{\klb}[1]{\todo[color={green!20}, inline]{KLB: #1}}

\usepackage[nameinlink]{cleveref}	
\creflabelformat{equation}{#1#2#3}	
\crefname{equation}{eq.}{eqs.}	
\Crefname{equation}{Eq.}{Eqs.}	
\Crefname{section}{\S}{\S}

\DeclareMathOperator*{\argmax}{arg\,max}
\newtheorem{theorem}{Theorem}

\newcommand{\allinst}{\mathcal{Z}}
\newcommand{\valid}{\mathcal{V}}
\newcommand{\invalid}{\mathcal{I}}

\newcommand{\lb}{\hat{l}}
\newcommand{\ub}{\hat{u}}
\newcommand{\modalint}{\hat{\mathcal{I}}_{\mathrm{mode}}}
\newcommand{\wald}[1]{\hat{\beta}_{#1}}

\title{Valid Causal Inference \\with (Some) Invalid Instruments}

\author{
Jason Hartford\\
Department of Computer Science\\
University of British Columbia\\
\texttt{jasonhar@cs.ubc.ca} \\
\And
Victor Veitch\\
Google\\ 
\texttt{victorveitch@gmail.com} \\
\And
Dhanya Sridhar\\
Data Science Institute\\
Columbia University\\
\texttt{dhanya.sridhar@columbia.edu} \\
\And
Kevin Leyton-Brown\\
Department of Computer Science\\
University of British Columbia\\
\texttt{kevinlb@cs.ubc.ca} \\
}

\begin{document}

\maketitle

\begin{abstract}
Instrumental variable methods provide a powerful approach to estimating causal effects in the presence of unobserved confounding. But a key challenge when applying them is the reliance on untestable ``exclusion'' assumptions that rule out any relationship between the instrument variable and the response that is not mediated by the treatment. In this paper, we show how to perform consistent IV estimation despite violations of the exclusion assumption. In particular, we show that when one has multiple candidate instruments, only a majority of these candidates---or, more generally, the modal candidate--response relationship---needs to be valid to estimate the causal effect. Our approach uses an estimate of the modal prediction from an ensemble of instrumental variable estimators. The technique is simple to apply and is ``black-box'' in the sense that it may be used with any instrumental variable estimator as long as the treatment effect is identified for each valid instrument independently. As such, it is compatible with recent machine-learning based estimators that allow for the estimation of conditional average treatment effects (CATE) on complex, high dimensional data. Experimentally, we achieve accurate estimates of conditional average treatment effects using an ensemble of deep network-based estimators, including on a challenging simulated Mendelian Randomization problem.
\end{abstract}

\section{Introduction}


Instrumental variable (IV) methods are a powerful approach for estimating treatment effects: they are robust to unobserved confounders and they are compatible with a variety of flexible nonlinear function approximators \citep[see e.g.][]{newey2003instrumental, darolles2011nonparametric, hartford2017deep, lewis2018adversarial, Singh2019, bennett2019deep}, thereby allowing nonlinear estimation of heterogeneous treatment effects.

In order to use an IV approach, one must make three assumptions. The first, \emph{relevance}, asserts that the treatment is not independent of the instrument. This assumption is relatively unproblematic, because it can be verified with data. The second assumption, \emph{unconfounded instrument}, asserts that the instrument and outcome do not share any common causes. This assumption cannot be verified directly, but in some cases it can be justified via knowledge of the system; e.g. the instrument may be explicitly randomized or may be the result of some well understood random process. The final assumption, \emph{exclusion}, asserts that the instrument's effect on the outcome is entirely mediated through the treatment. This assumption is even more problematic; not only can it not be verified directly, but it can be very difficult to rule out the possibility of direct effects between the instrument and
the outcome variable. 
Indeed, there are prominent cases where purported instruments have been called into question for this reason. For example, in economics, the widely used ``judge fixed effects" research design \citep{kling2006incarceration} uses random assignment of trial judges as instruments and leverages differences between different judges' propensities to incarcerate to infer the effect of incarceration on some economic outcome of interest \citep[see][for many recent examples]{frandsen2019judging}. \citet{mueller2015criminal} points out that exclusion is violated if judges also hand out other forms of punishment (e.g. fines, a stern verbal warning etc.) that are not observed. Similarly, in genetic epidemiology ``Mendelian randomization" \citep{davey2003mendelian}  uses genetic variation to study the effects of some exposure on an outcome of interest. For example, given genetic markers that are known to be associated with a higher body mass index (BMI), we can estimate the effect of BMI on cardiovascular disease. However, this only holds if we are confident that the same genetic markers do not influence the risk of cardiovascular disease in any other ways. The possibility of such ``direct effects''---referred to as ``horizontal pleiotropy" in the genetic epidemiology literature---is regarded as a key challenge for Mendelian randomization \citep{hemani2018evaluating}.

It is sometimes possible to identify \emph{many} candidate instruments, each of which satisfies the relevance assumption; in such settings, demonstrating exclusion is usually the key challenge, though in principle unconfounded instrument could also be a challenge. For example, many such candidate instruments can be obtained in both the judge fixed effects and Mendelian randomization settings, where individual judges and genetic markers, respectively, are treated as different instruments. 
Rather than asking the modeler to gamble by choosing a single candidate about which to assert these untestable assumptions, this paper advocates making a weaker assumption about the whole set of candidates. Most intuitively, we can assume \emph{majority validity}: that at least a majority of the candidate instruments satisfy all three assumptions, even if we do not know which candidates are valid and which are invalid. Or we can go further and make the still weaker assumption of \emph{modal validity}: that the modal relationship between instruments and response is valid. 
Observe that modal validity is a weaker condition 
because if a majority of candidate instruments are valid, 
the modal candidate--response relationship must be 
characterized by these valid instruments. 
Modal validity is satisfied if, as Tolsoy might have said, ``All happy instruments are alike; each unhappy instrument is unhappy in its own way.''


This paper introduces ModeIV,
a robust instrumental variable technique that we show is asymptotically valid when modal validity holds. ModeIV allows the estimation of nonlinear causal effects and lets us estimate conditional average treatment effects that vary with observed covariates. ModeIV is a black-box method in the sense that it is compatible with any valid IV estimator, which allows it to leverage any of the recent machine learning-based IV estimators. 
%
We experimentally validated ModeIV using both a modified version of the \citet{hartford2017deep} demand simulation
and a more realistic Mendelian randomization example. In both settings---even when we generated data with a very low signal-to-noise ratio---we observed ModeIV to be robust to exclusion-restriction bias and to accurately recover conditional average treatment effects (CATE).

\section{Background on Instrumental Variables}
We are interested in estimating the causal effect of some treatment variable, $t$, on some outcome of interest, $y$. 
The treatment effect is confounded by a set of observed covariates, $x$, as well as unobserved confounding factors, $\epsilon$, which affect both $y$ and $t$. With unobserved confounding, we cannot rely on conditioning to remove the effect of confounders; instead we use an instrumental variable, $z$, to identify the causal effect.

Instrumental variable estimation can be thought of as an inverse problem: we can directly identify the causal\footnote{Strictly, non-causal instruments suffice but identification and interpretation of the estimates can be more subtle \citep[see][]{swanson2018challenging}.} effect of the instrument on both the treatment and the response before asking the inverse question, ``what treatment--response mappings, $f:t\rightarrow y$, could explain the difference between these two effects?" The problem is identified if this question has a unique answer. If the true structural relationship is of the form, $y = f(t, x) + \epsilon$, one can show that, 
$
  E[y|x,z] = \int f(t,x) d F(t |x, z), 
$
 where $E[y|x,z]$ gives the instrument--response relationship, $F(t | x,z)$ captures the instrument--treatment relationship, and the goal is to solve the inverse problem to find $f(\cdot)$. In the linear case, $f(t, x) = \beta t + \gamma x$, so the integral on the right hand side of reduces to $\beta E[t| x, z] + \gamma x$ and $\beta$ can be estimated using linear regression of $y$ on the predicted values of $t$ given $x$ and $z$ from a first stage regression. This procedure is known as Two-Stage Least Squares \citep[see][]{angrist2008mostly}. 
More generally, the causal effect is identified if the integral equation has a unique solution for $f$ \citep[for details, see][]{newey2003instrumental}.

A number of recent approaches have leveraged this additive confounders assumption to extend IV analysis beyond the linear setting.
\citet{newey2003instrumental} and \citet{darolles2011nonparametric} proposed the first nonparametric procedures for estimating these structural equations, based on polynomial basis expansion. These methods relax the linearity requirement, but scale poorly in both the number of data points and the dimensionality of the data. 
To overcome these limitations, recent approaches have adapted deep neural networks for nonlinear IV analyses. 
DeepIV \citep{hartford2017deep} 
fits a first-stage conditional density estimate of $\hat{F}(t | x, z)$ and uses it to solve the above integral equation. 
 Both \citet{lewis2018adversarial} and \citet{bennett2019deep} adapt generalized method of moments \citep{hansen1982} to the nonlinear setting by leveraging adversarial losses, while \citet{Singh2019} proposed a kernel-based approach. 
 \citet{Puli2019} showed that constraints on the structural equation for the treatment also lead to identification.

\section{Related work}

\paragraph{Inference with invalid instruments in linear settings} 
Much of the work on valid inference with invalid instruments is in the Mendelian Randomization literature, where violations of the exclusion restriction are common. For a recent survey see \citet{hemani2018evaluating}. 
There are two broad approaches to valid inference in the presence of bias introduced by invalid instruments: averaging over the bias, or eliminating the bias with ideas from robust statistics.
In the first setting, valid inference is possible under the assumption that each instrument introduces a random bias, but that the mean of this process if zero (although this assumption can be relaxed, c.f. \citet{Bowden2015, kolesar2015identification}). Then, the bias tends to zero as the number of instruments grow. 
Methods in this first broad class have the attractive property that they remain valid even if none of the instruments is valid, but they rely on strong assumptions that do not easily generalize to the nonlinear setting considered in this paper.
The second class of approaches to valid inference assumes that some fraction of the instruments are valid. Then, biased instruments are outliers whose effect can be removed by leveraging the robustness of the median \citep{kang2016instrumental} and the mode \citep{hartwig2017robust}.
In this paper, we use the same robust
statistics insights, generalizing them to the nonlinear setting.

\paragraph{Ensemble models} Ensembles are widely used in machine learning as a technique for improving prediction performance by reducing variance \citep{Breiman1996} and combining the predictions of weak learners trained on non-uniformly sampled data \citep{freund1995desicion}. These ensemble methods frequently use modal predictions via majority voting among classifiers, but they are designed to reduce variance.
Both the median and mode of an ensemble of models have been explored as a way of improve robustness to outliers in the forecasting literature \citep{stock_combination_2004, kourentzes_neural_2014}, but we are not aware of any prior work that explicitly uses these aggregation techniques to eliminate bias from an ensemble.

\paragraph{Mode estimation}
If a distribution admits a density, the mode is defined as the global maximum of the density function. More generally,  the mode can be defined as the limit of a sequence of modal intervals---intervals of width $h$ that contains the largest proportion of probability mass---such that $x_{\text{mode}}=\lim_{h\rightarrow 0}\argmax_x F([x-h/2, x+h/2])$. 
These two definitions suggest two estimation methods for estimating the mode from samples: either one may try to estimate the density function and the maximize the estimated function \citep{parzen_estimation_1962}, or one might search for midpoints of modal intervals from the empirical distribution functions.  To find modal intervals, one can either fix an interval width, $h$, and choose $x$ to maximize the number of samples within the modal interval \citep{chernoff_estimation_1964}, or one can solve the dual problem by fixing the target number of samples to fall into the modal interval and minimizing $h$ \citep{dalenius_mode--neglected_1965, venter_estimation_1967}. We use this latter \citeauthor{dalenius_mode--neglected_1965}--\citeauthor{venter_estimation_1967} approach as the target number of samples can be parameterized in terms of the number of valid instruments, thereby avoiding the need to select a kernel bandwidth $h$.

%




\section{ModeIV}
\label{sec:methods}

In this paper, we assume we have access to a set of $k$ independent\footnote{This independence requirement can be weakened to conditionally independent, conditional on some variable $c$. See \cref{sec:relax} for details.} candidate instrumental variables, $\mathcal{Z}=\{z_1, \dots, z_k\}$, which are `valid' if they satisfy relevance, exclusion and unconfounded instrument, and `invalid' otherwise. 
Denote the set of valid instruments, $\valid:=\{z_i: z_i \perp y | x,t, \epsilon \}$, and the set of invalid instruments, $\invalid = \allinst \setminus \valid$. 
We further assume that each valid instrument identifies the causal effect. This amounts to assuming that the unobserved confounder's affect on $y$ is additive, $y = f(t, x, z_{i:i\in \invalid}) + \epsilon$ for some function $f$ and $E[y|x,z_{i:i\neq j}, z_j] = \int f(t,x,z_{i:i\neq j}) d F(t |x,z_{i:i\neq j}, z_j)$ has a unique solution for all $j$ in $\allinst$.


The ModeIV procedure requires the analyst to specify a lower bound $V \geq 2$ on the number of valid instruments and then proceeds in three steps.
\begin{enumerate}
\item Fit an ensemble of $k$ estimates of the conditional outcome $\{\hat{f}_1, \dots, \hat{f}_k\}$ using a non-linear IV procedure applied to each of the $k$ instruments.
  Each $\hat{f}$ is a function mapping treatment $t$ and covariates $x$ to an estimate of the effect of the treatment conditional on $x$.
\item For a given test point $(t,x)$,
  select $[\lb, \ub]$ as a smallest interval containing $V$ of the estimates $\{\hat{f}_1(t,x), \dots, \hat{f}_k(t,x)\}$.
  Define $\modalint = \{i : \lb < \hat{f}_{i}(t,x) < \ub\}$ to be the indices of the instruments corresponding to estimates falling in the interval.
\item Return $\hat{f}_{\mathrm{mode}}(t,x) = \frac{1}{\lvert{\modalint}\rvert}\sum_{i \in \modalint} \hat{f}_i(t,x)$
\end{enumerate}

The idea is that the estimates from the valid instruments will tend to cluster around the true value of the effect,
$E[y|\text{do}(t), x]$.
We assume that the most common effect instrument is a valid one; i.e., that the modal effect is valid.
To estimate the mode, we look for the tightest cluster of points. These are the points contained in $\modalint$.
Intuitively, each estimate in this interval should be approximately valid and hence approximates the modal effect.
We take the average of these estimates to gain statistical strength.

The next theorem formalizes this intuition 
and shows, in particular, that ModeIV asymptotically identifies
and consistently estimates the causal effect.
\begin{theorem}
  Fix a test point $(t,x)$ and let $\wald{1}, \dots, \wald{k}$ be estimators of the causal effect of $t$ at $x$ corresponding to $k$ (possibly invalid) instruments. 
  E.g., $\wald{j} = \hat{f}_j(t,x)$.
  Denote the true effect as $\beta = E[y|\text{do}(t), x]$.
  Suppose that
  \begin{enumerate}
  \item  (consistent estimators) $\wald{j} \to \beta_j$ almost surely for each instrument. In particular, $\beta_j = \beta$ whenever the $j$th instrument is valid.
  \item (modal validity) At least $p\%$ of the instruments are valid, and no more than $p\% - \epsilon$ of the invalid instruments agree on an effect. That is, $p\%$ of the instruments yield the same estimand if and only if all of those instruments are valid.
  \klb{Is the last statement true? Couldn't it be more than $p\%$ if some invalid instruments agree on the effect too? (This violates the ``only if''.)}
  \end{enumerate}
  Let $[\lb, \ub]$ be the smallest interval containing $p\%$ of the instruments and let $\modalint = \{i : \lb < \wald{i} < \ub\}$.
  Then,
  \[
    \sum_{i \in \modalint} \hat{w}_i \wald{i} \to \beta
\]
almost surely, where $\hat{w}_i,w_i$ are any non-negative set of weights such that each $\hat{w}_i \to w_i$ a.s. and $\sum_{i \in \modalint} w_i = 1$.
Further suppose that the individual estimators are also asymptotically normal,
\[
  \sqrt{n}\wald{j} \to N(\beta_j,\sigma_j^2)
\]
for each instrument. Then it also holds that the modal estimator is asymptotically normal:
  \[
   \sqrt{n} \sum_{i \in \modalint} \hat{w}_i \wald{i} \to N(\beta, \sum_iw_i^2\sigma_i^2).
 \]
\label{thm:consistent}
\end{theorem}
We defer the proof to the supplementary material.

Of course, the ModeIV procedure can be generalized to allow different estimators of the mode than
the one used in Steps 2 and 3.
The particular choice we make here has the advantage of being straightforward,
statistically stable, computationally inexpensive, and relatively insensitive to the choice of $V$. 
The procedure as a whole is, however, $k$ times more computationally expensive than running single estimation procedure at both training and test time.

\section{Experiments}

\label{sec:exp}

We studied ModeIV empirically in two simulation settings. First, we investigated the performance of ModeIV for non-linear effect estimation as the proportion of invalid instruments increased for various amounts of direct effect bias. Second, we applied ModeIV to a realistic Mendelian randomization simulation to estimate heterogeneous treatment effects. For all experiments, we use DeepIV \citep{hartford2017deep} as the nonlinear estimator. Full experimental details are given in the appendix.

\subsection{Biased demand simulation}
\label{sec:exp_demand}

We evaluated the effect of invalid instruments on estimation by modifying the low dimensional demand simulation from \citet{hartford2017deep} to include multiple candidate instruments. The \citeauthor{hartford2017deep} demand simulation models a scenario where the treatment effect varies as a nonlinear function\footnote{$\psi(t)=2\left((t-5)^{4} / 600+\exp \left[-4(t-5)^{2}\right]+t / 10-2\right)$. See the appendix for a plot of the function} of time $\psi(t)$, and observed covariates $x$.
\begin{align*}
    \textcolor{Maroon}{z_{1:k}}, \nu \sim \mathcal{N}(0, 1) \quad t\sim \text{unif}(0, 10) \quad &e\sim\mathcal{N}(\rho \nu, 1 - \rho^2),\qquad  p = 25 + (\textcolor{Maroon}{z^T \beta^{(zp)}} + 3)\psi(t) + \nu \\
    y = 100 + 10x^T \beta^{(x)} \psi(t)\, + &\underbrace{(x^T \beta^{(x)} \psi(t) - 2)}_{\text{Treatment effect}}p\, + \underbrace{\textcolor{Maroon}{\gamma \sin(z^T \beta^{(zy)})}}_{\text{Exclusion violation}}  +\, e
\end{align*}
We highlight the differences between our data generating process and the \citeauthor{hartford2017deep} data generating process in red: we have $k$ instruments whose effect on the treatment is parameterized by $\beta^{(zx)}$, instead of a single instrument in the original; we include an exclusion violation term which introduces bias into standard IV approaches whenever $\gamma$ is non-zero. The vector $\beta^{(zy)}$ controls the direct effect of each instrument: invalid instruments have nonzero $\beta^{(zy)}_i$ coefficients, while valid instrument coefficients are zero.

\begin{figure}[t]
  \begin{center}
  \includegraphics{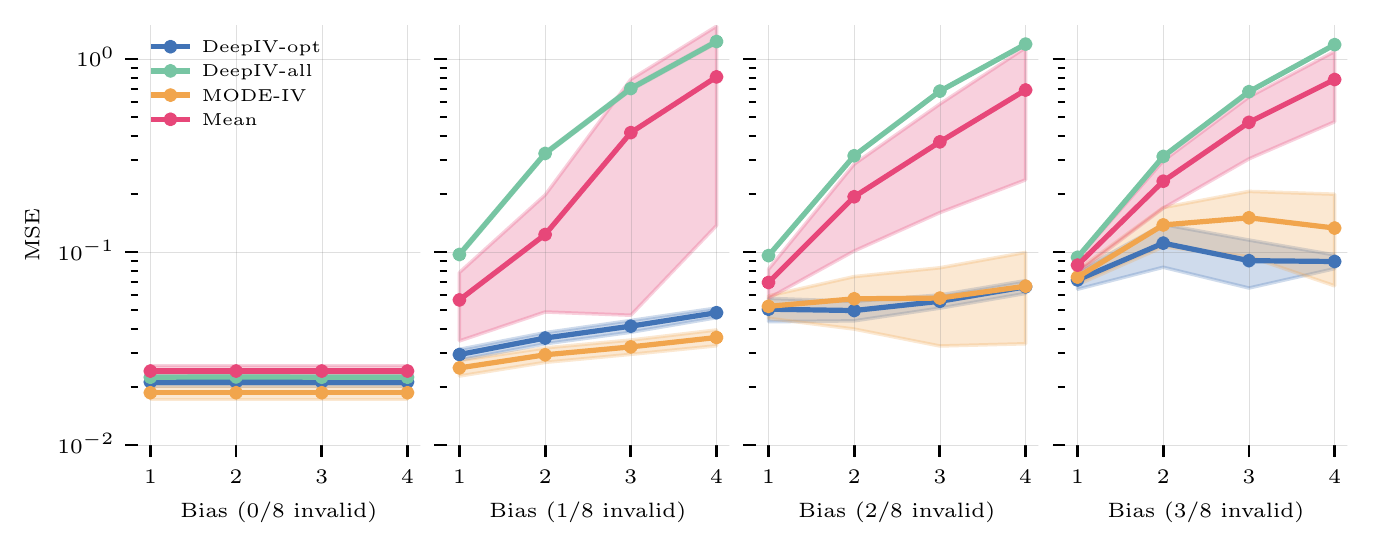}
  \caption{Performance on the biased demand simulation for various numbers of invalid instruments. The $x$-axis shows, $\gamma$, the scaling factor that scales the amount of exclusion violation bias.}
  \label{fig:bias}
  \end{center}
\end{figure}

 We fitted an ensemble of $k$ different DeepIV models that were each trained with a different instrument $z_i$. In Figure \ref{fig:bias}, we compare the performance of ModeIV with three baselines: DeepIV with oracle access to the set of valid instruments (DeepIV-opt); 
the ensemble mean (Mean); and a naive approach that fit a single instance of DeepIV treating all instruments as valid (DeepIV-all). The $x$-axis of the plots indicates the scaling factor $\gamma$, which scales the amount of bias introduced via violations of the exclusion restriction. 

All methods performed well when all the instruments were valid. 
Once the methods had to contend with invalid instruments, Mean and DeepIV-all performed significantly worse than ModeIV because of both methods' sensitivity to the biased instruments. ModeIV's mean squared error closely tracked that of the oracle method as the number of biased instruments increased, and the raw mean squared errors of both methods also increased as the number of valid instruments in the respective ensembles correspondingly fell. 

\begin{figure}[t]
  \begin{center}
  \includegraphics{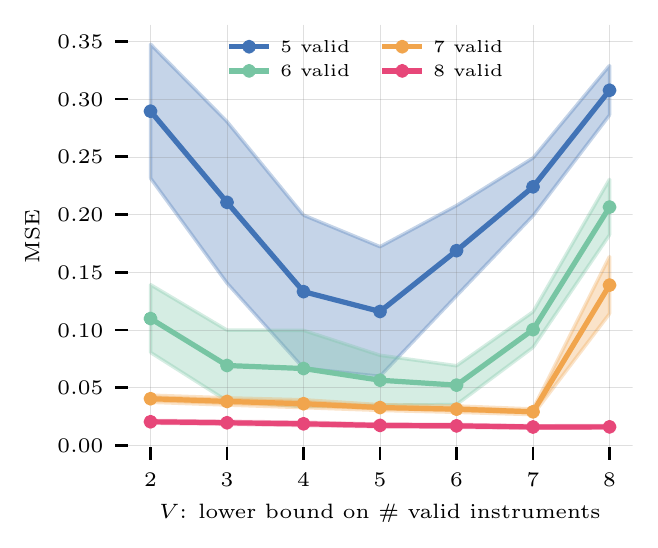}
  \includegraphics{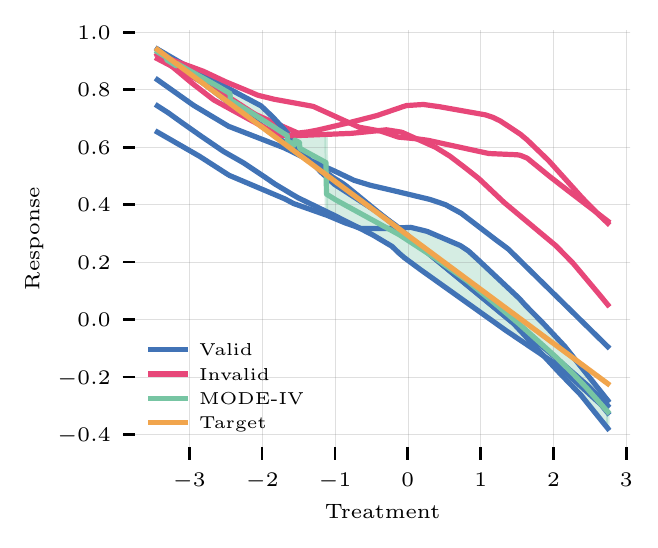}
  \caption{\emph{(Left)} ModeIV's sensitivity to the choice of number of valid instruments parameter $V$. Best performance is achieve when $V$ is equal to the true number of valid instruments, the method is relatively insensitive to more conservative choices of $V$. 
  \emph{(Right)} Example of the ModeIV algorithm on a biased demand simulation with 5 valid instruments and 3 invalid instruments. For the plot we fixed an arbitrary value of $x$ and $z$ and varied $t$. The region highlighted in green contains the 4 predictions that formed part of the modal interval for each given input $t'$. 
  The modal prediction is shown in solid green.
  }
  \label{fig:sens}
  \end{center}
\end{figure}

\paragraph{Sensitivity}
\label{sec:sensitivity}
When using ModeIV, one key practical question that an analyst faces is choosing $V$, the lower bound on the number of valid instruments. We evaluated the importance of this choice in figure \ref{fig:sens} (\emph{left}) by testing the performance of ModeIV across the full range of choices for $V$ with different numbers of biased instruments. We found that, as expected, the best performance was achieved when $V$ equaled the true number of valid instruments, but also that similar levels of performance could be achieved with more conservative choices of $V$. That said, with only 5 valid instruments, ModeIV tended to perform worse when $V$ was set  too small. To see why this is the case, notice that in figure \ref{fig:sens} (\emph{right}), there are a number of regions of the input space where the invalid instruments agreed by chance (e.g. $t\in [-1, 0.]$), so these regions bias ModeIV for small mode set sizes. Overall, we observed that setting $V$ to half the number of instruments tended to work well in practice.

Asymptotically, ModeIV remains consistent when fewer than half of the instruments are valid, but when this is the case there are far more ways that Assumption 2 of Theorem \ref{thm:consistent} can be violated. 
This is illustrated in \autoref{fig:sens} (\emph{right}) which shows that there are a number of regions where the bias instruments agree by chance. Because of this, we recommend only using ModeIV when one can assume that the majority of instruments are valid, unless one has prior knowledge to justify \emph{modal validity} without assuming the majority of instruments are valid\footnote{For example, if direct effects are strictly monotone and disagree, chance agreements among invalid instruments can only occur in a finite number of locations.}.

\paragraph{Selecting instruments?}
ModeIV constitutes a consistent method for making unbiased predictions but, somewhat counter-intuitively, it does not directly offer a way of inferring the set of valid instruments. For example, one might imagine identifying the set of candidates that most often form part of the modal interval $\modalint$.
The problem is that while candidates that fall within the modal interval $\modalint$ tend to be close to the mode, the interval can include invalid instruments that yielded an effect close to the mode by chance. \klb{Why does it have to be by chance? Couldn't it also be because an adversary set it up that way? I.e., we don't have a model of the invalid instruments as having to be random, do we? The paper is written in a way that seems to equivocate on this point; I think the real result is robust to adversarial manipulation, but much of the discussion is written as though invalid instruments must be uncorrelated.}
\klb{But if it is by chance, give the intuition about why the same invalid instrument can be chosen across most $x$ values. (I think this requires being more formal about what's meant by "by chance".)}
Since these invalid estimates are close to the truth, they do not hurt the estimate.
We can see this in \autoref{fig:sens} (\emph{right}) where invalid instruments form part of the modal interval in the region $t\in [-3.5, -2]$, without introducing bias.
\klb{I see that you've put this paragraph here because the last bit refers to the figure. But in general, this feels much more like discussion. It's tempting to move this to Section 6, though I guess then Figure 2 (right) would be unexplained unless it were moved to Section 6 also.}

\subsection{Mendelian randomization simulation}
\label{sec:exp_mr}

\begin{figure}[t]
  \begin{center}
  \includegraphics[]{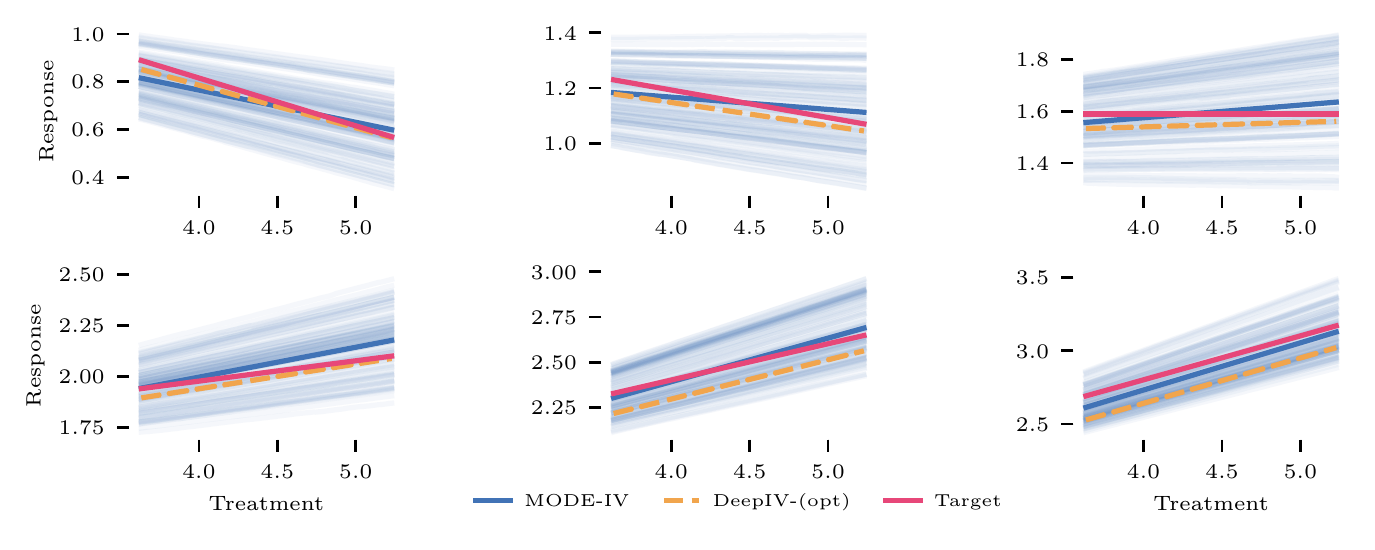}
  \caption{Estimated conditional dose--response curves for the Mendelian randomization simulation. Each light blue curve shows ModeIV's estimate $f(t, x)$ for some IID sample of $x$; each figure's dark curve represents the average over all samples of $x$. The six plots show the different subsets of the range, $x$, where true slope $\beta(x)$ is \emph{(left to right)} $-0.2, -0.1, 0., 0.1, 0.2$ and $0.3$ respectively. 
  }
  \label{fig:predicted}
  \end{center}
\end{figure}


We evaluate our approach on the simulated data adapted from \citet{hartwig2017robust}, which is designed to reflect violations of the exclusion restriction in Mendelian randomization studies. 

Instruments, $z_i$, represent SNPs---locations in the genetic sequence where there is frequent variation among people---modeled as random variables drawn from a Binomial$(2, p_i)$ distribution corresponding to the frequency with which an individual gets one or both rare genetic variants. The treatment and response are both continuous functions of the instruments with Gaussian error terms. The strength of the instrument's effect on the treatment, $\alpha_i$, and direct effect on the response, $\delta_i$, are both drawn from Uniform$(0.01, 0.2)$ distributions for all $i$. For all experiments we used 100 candidate instruments and varied the number of valid instruments from 50 to 100 in increments of 10; we set $\delta_i$ to 0 for all valid instruments. More formally,
\begin{align*}
    &z_i \sim \text{Binomial}(2, p_i) \quad \text{for $i$ in }[1\dots K],\quad \beta(x) := \text{round}(x^T\gamma^{(xt)}, 0.1).\\ 
    &t := \sum_{j=1}^K \alpha_j z_j + \rho u + \epsilon_x \quad y :=\beta(x) t + \sum_{j=1}^K \delta_j z_j + u + \epsilon_y  \nonumber
\end{align*}

In the original \citeauthor{hartwig2017robust} simulation, the treatment effect $\beta$ was fixed for all individuals. Here, we make the treatment effect vary as a function of observable characteristics to model a scenario where treatments may affect different sub-populations differently. We simulate this by making the treatment effect, $\beta(x)$, a sparse linear function of observable characteristics, $x\in R^{10}$, where 3 of the 10 coefficents, $\gamma_i^{(xt)}$ were sampled from $U(0.2, 0.5)$ and the remaining $\gamma_i^{(xt)}$ were set to $0$. We introduce non-linearity by rounding to the nearest 0.1, 
which makes the learning problem harder, while making it easier to visually show the differences between the fitted functions and their true targets.\\

Mendelian randomization problems tend to have low signal-to-noise ratios, where the treatment explains only 1-3\% of the response variance.\klb{Support this with a citation?} 
This makes the setting challenging for neural networks, which tend to perform best on low-noise regimes.\klb{Again, a cite, or an argument? I've certainly seen neural networks do great in the presence of lots of noise.} To address this, we leveraged the inductive bias that the data is conditionally linear in the treatment effect, using a neural network to parameterize the slope of the treatment variable rather than outputting the response directly. So, for these problems, we defined $\hat{f}(t,x) = g(\phi(x)) t + h(\phi(x))$, where $g(\cdot)$ and $h(\cdot)$ are linear layers that act on a shared representation $\phi(x)$.

The general trends we observed on the Mendelian randomization 
simulation, summarized in Table \ref{tab:mse}, were similar to those we observed in the biased demand simulation: DeepIV-all performed poorly; ModeIV closely tracked the performance of our oracle, DeepIV-opt. On this simulation the mean ensemble Mean achieved stronger performance, but still significantly underperformed ModeIV. 

\paragraph{Conditional average treatment effects} Figure \ref{fig:predicted} shows the predicted dose--response curves for a variety of different levels of the true treatment effect. The six plots correspond to six different subspaces of $x$ which all have the same true conditional treatment effect. Each of the light blue lines shows ModeIV's prediction for a different value of $x$. 
The model is not told that the true $\beta$ is constant for each of these sub-regions, but instead has to learn that from data so there is some variation in the slope of each prediction. Despite this, the majority of predicted curves match the sign of the treatment effect and closely match the ground truth slope. 
The average absolute bias in ModeIV's conditional average treatment effect estimation was 0.01 larger than that of DeepIV on the valid instruments, for true effect sizes that range between -0.3 and 0.3 (see Table \ref{tab:grad_bias} in the appendix for details).

\begin{table}
\centering
\resizebox{\columnwidth}{!}{%
\begin{tabular}{lllllll}
\toprule
Model &                         50\% valid  &                         60\% valid  &                         70\% valid  &                         80\% valid  &                         90\% valid  &                         100\% valid \\
\midrule
DeepIV (valid) &  0.035 \tiny{$\pm$ (0.001)} &  0.035 \tiny{$\pm$ (0.001)} &  0.034 \tiny{$\pm$ (0.001)} &  0.034 \tiny{$\pm$ (0.001)} &    0.032 \tiny{$\pm$ (0.0)} &  0.024 \tiny{$\pm$ (0.001)} \\
ModeIV 30\%   &  0.037 \tiny{$\pm$ (0.001)} &  0.037 \tiny{$\pm$ (0.001)} &  0.038 \tiny{$\pm$ (0.001)} &  0.039 \tiny{$\pm$ (0.001)} &  0.041 \tiny{$\pm$ (0.001)} &  0.032 \tiny{$\pm$ (0.001)} \\
ModeIV 50\%   &  0.037 \tiny{$\pm$ (0.001)} &  0.036 \tiny{$\pm$ (0.001)} &  0.038 \tiny{$\pm$ (0.001)} &  0.039 \tiny{$\pm$ (0.001)} &   0.04 \tiny{$\pm$ (0.001)} &  0.032 \tiny{$\pm$ (0.001)} \\
Mean Ensemble           &  0.041 \tiny{$\pm$ (0.001)} &  0.041 \tiny{$\pm$ (0.001)} &  0.043 \tiny{$\pm$ (0.001)} &  0.043 \tiny{$\pm$ (0.001)} &  0.045 \tiny{$\pm$ (0.001)} &  0.036 \tiny{$\pm$ (0.001)} \\
DeepIV (all)   &  0.099 \tiny{$\pm$ (0.007)} &  0.116 \tiny{$\pm$ (0.005)} &  0.149 \tiny{$\pm$ (0.005)} &  0.149 \tiny{$\pm$ (0.005)} &  0.142 \tiny{$\pm$ (0.003)} &    0.025 \tiny{$\pm$ (0.0)} \\
\bottomrule
\end{tabular}%
}
\vspace{.4em}
\caption{Performance on the Mendelian randomization simulation for various proportions of valid instruments. The ensemble methods performed far better than the DeepIV model which treated all instruments as valid, and ModeIV gave significantly better performance than the mean ensemble, was close to the performance of DeepIV on the valid instruments.}
\label{tab:mse}
\end{table}




\section{Discussion and Limitations}
\label{sec:discussion}
The conventional wisdom for IV analysis is: if you have many (strong) instruments and sufficient data, you should use all of them so that your estimator can maximize statistical efficiency by weighting the instruments appropriately. This remains true in our setting---indeed, DeepIV trained on the valid instruments typically outperformed any of the ensemble techniques---but of course requires a procedure for identifying the set of valid instruments. In the absence of such a procedure, falsely assuming that all candidate instruments are valid can lead to large biases, as illustrated by the poor performance of DeepIV-all on simulations that included bias. ModeIV gives up some efficiency by filtering instruments, but it gains robustness to invalid instruments and in practice we found that the loss of efficiency was negligible. Of course, that empirical finding will vary across settings. A useful future direction would find a procedure for recovering the set of valid instruments to further reduce the efficiency trade-offs. \klb{This is where I'd put the earlier discussion and experimental results about whether we can identify the valid instruments.}

There are, however, some important settings where ModeIV either will not work or require more careful assumptions. First, 
our key assumption was that each valid instrument consistently estimates the same function, $f(t,x)$. In settings with discrete treatments, one typically only identifies a ``(conditional) local average treatment effect'' (CLATE / LATE respectively) for each instrument. The LATE for instrument $i$ can be thought of as the average treatment effect for the sub-population that change their behavior in response to a change in the value of instrument $i$; if the LATEs differ across instruments, this implies that each instrument will result in a different estimate of $E[\hat{f}_i(t,x)]$ regardless of whether any of the instruments are invalid. In such settings, ModeIV will return the average of the $V$ closest $\hat{f}_i(t,x)$'s, but one would need additional assumptions on how these estimates cluster relative to biased estimates to apply any causal interpretation to this quantity.
%
The alternative is the approach that we take here: assume that a common function $f(t,x)$ is shared across all units and allow for heterogeneous treatment effects by allowing the treatment effect to vary as a function of observed covariates $x$. This shared heterogeneous effect assumption is weaker than prior work on robust IV, which requires a ``constant effect'' effect assumption that every individual responds in exactly the same way to the treatment via a linear parameter, $\beta$.

Second, we have focused on settings where each instrument is independent. 
In principle ModeIV extends to settings where some instruments are confounded or all the instruments share a common cause, but the conditions for valid inference are more delicate because one has to ensure all backdoor paths between the instruments and response are blocked (see \cref{sec:relax}).

\section*{Broader Impact}
IV methods are important tools in the causal inference toolbox. These methods have been applied to study causal effects across a wide range of settings, spanning economic policies, phenotypes that may cause disease, and recommendation algorithms. Typically, analysts use their best expert judgment to assess whether various candidate instruments satisfy the exclusion restriction and then proceed with effect estimation. However, this judgement is both difficult to make and highly consequential: biased instruments can invalidate the analyst's conclusions. 

This paper offers the analyst a potentially easier alternative: assuming that some fixed proportion of instruments are valid.
With this, the paper provides a method for estimating causal effects in the presence of invalid instruments, backed by theoretical guarantees. Because ModeIV captures nonlinear effects, it may be better suited to applications in genetics and economics than previous methods.

Of course, caveats still apply. As with all causal inference methods, ModeIV still requires an analyst to make assumptions and to assess potentially delicate, mathematical statements. Although we have striven to make the limitations of our method explicit and to provide guidelines where possible, negative impacts could arise if analysts apply our method without prudence and thus obtain invalid causal conclusions, particularly if these are offered as policy suggestions. To mitigate the potential for such negative impacts, we have clearly noted when ModeIV should not be applied. Beyond that, we recommend that analysts considering applying this or any other causal estimation procedure conduct as much sensitivity analysis and external validation as possible.



\begin{ack}
This work was supported by Compute Canada, a GPU grant from NVIDIA, an NSERC Discovery Grant, a DND/NSERC Discovery Grant Supplement, a CIFAR Canada AI Research Chair at the Alberta Machine Intelligence Institute, and DARPA award FA8750-19-2-0222, CFDA\# 12.910, sponsored by the Air Force Research Laboratory.
\end{ack}

\bibliography{refs}

\begin{thebibliography}{29}
\providecommand{\natexlab}[1]{#1}
\providecommand{\url}[1]{\texttt{#1}}
\expandafter\ifx\csname urlstyle\endcsname\relax
  \providecommand{\doi}[1]{doi: #1}\else
  \providecommand{\doi}{doi: \begingroup \urlstyle{rm}\Url}\fi

\bibitem[Angrist and Pischke(2008)]{angrist2008mostly}
J.~D. Angrist and J.-S. Pischke.
\newblock \emph{Mostly harmless econometrics: An empiricist's companion}.
\newblock Princeton university press, 2008.

\bibitem[Bennett et~al.(2019)Bennett, Kallus, and Schnabel]{bennett2019deep}
A.~Bennett, N.~Kallus, and T.~Schnabel.
\newblock Deep generalized method of moments for instrumental variable
  analysis.
\newblock \emph{arXiv preprint arXiv:1905.12495}, 2019.
\newblock URL \url{https://arxiv.org/abs/1905.12495}.

\bibitem[Bowden et~al.(2015)Bowden, Davey~Smith, and Burgess]{Bowden2015}
J.~Bowden, G.~Davey~Smith, and S.~Burgess.
\newblock {Mendelian randomization with invalid instruments: effect estimation
  and bias detection through Egger regression}.
\newblock \emph{International Journal of Epidemiology}, 44\penalty0
  (2):\penalty0 512--525, 06 2015.

\bibitem[Breiman(1996)]{Breiman1996}
L.~Breiman.
\newblock Bagging predictors.
\newblock \emph{Machine Learning}, 24\penalty0 (2):\penalty0 123--140, 1996.

\bibitem[Chernoff(1964)]{chernoff_estimation_1964}
H.~Chernoff.
\newblock Estimation of the mode.
\newblock \emph{Annals of the Institute of Statistical Mathematics},
  16\penalty0 (1):\penalty0 31--41, Dec. 1964.

\bibitem[Dalenius(1965)]{dalenius_mode--neglected_1965}
T.~Dalenius.
\newblock The {Mode}--{A} {Neglected} {Statistical} {Parameter}.
\newblock \emph{Journal of the Royal Statistical Society. Series A (General)},
  128\penalty0 (1):\penalty0 110, 1965.

\bibitem[Darolles et~al.(2011)Darolles, Fan, Florens, and
  Renault]{darolles2011nonparametric}
S.~Darolles, Y.~Fan, J.-P. Florens, and E.~Renault.
\newblock Nonparametric instrumental regression.
\newblock \emph{Econometrica}, 79\penalty0 (5):\penalty0 1541--1565, 2011.

\bibitem[Davey~Smith and Ebrahim(2003)]{davey2003mendelian}
G.~Davey~Smith and S.~Ebrahim.
\newblock ‘{M}endelian randomization’: can genetic epidemiology contribute
  to understanding environmental determinants of disease?
\newblock \emph{International journal of epidemiology}, 32\penalty0
  (1):\penalty0 1--22, 2003.

\bibitem[Frandsen et~al.(2019)Frandsen, Lefgren, and
  Leslie]{frandsen2019judging}
B.~R. Frandsen, L.~J. Lefgren, and E.~C. Leslie.
\newblock Judging judge fixed effects.
\newblock Technical report, National Bureau of Economic Research, 2019.

\bibitem[Freund and Schapire(1995)]{freund1995desicion}
Y.~Freund and R.~E. Schapire.
\newblock A desicion-theoretic generalization of on-line learning and an
  application to boosting.
\newblock In \emph{European conference on computational learning theory}, pages
  23--37. Springer, 1995.

\bibitem[Hansen(1982)]{hansen1982}
L.~P. Hansen.
\newblock Large sample properties of generalized method of moments estimators.
\newblock \emph{Econometrica}, 50\penalty0 (4):\penalty0 1029--1054, 1982.

\bibitem[Hartford et~al.(2017)Hartford, Lewis, Leyton-Brown, and
  Taddy]{hartford2017deep}
J.~Hartford, G.~Lewis, K.~Leyton-Brown, and M.~Taddy.
\newblock Deep {IV}: A flexible approach for counterfactual prediction.
\newblock In \emph{Proceedings of the 34th International Conference on Machine
  Learning-Volume 70}, pages 1414--1423. JMLR. org, 2017.

\bibitem[Hartwig et~al.(2017)Hartwig, Davey~Smith, and
  Bowden]{hartwig2017robust}
F.~P. Hartwig, G.~Davey~Smith, and J.~Bowden.
\newblock Robust inference in summary data {M}endelian randomization via the
  zero modal pleiotropy assumption.
\newblock \emph{International journal of epidemiology}, 46\penalty0
  (6):\penalty0 1985--1998, 2017.

\bibitem[Hemani et~al.(2018)Hemani, Bowden, and
  Davey~Smith]{hemani2018evaluating}
G.~Hemani, J.~Bowden, and G.~Davey~Smith.
\newblock Evaluating the potential role of pleiotropy in {M}endelian
  randomization studies.
\newblock \emph{Human molecular genetics}, 27\penalty0 (2):\penalty0 195--208,
  2018.

\bibitem[Kang et~al.(2016)Kang, Zhang, Cai, and Small]{kang2016instrumental}
H.~Kang, A.~Zhang, T.~T. Cai, and D.~S. Small.
\newblock Instrumental variables estimation with some invalid instruments and
  its application to {M}endelian randomization.
\newblock \emph{Journal of the American statistical Association}, 111\penalty0
  (513):\penalty0 132--144, 2016.

\bibitem[Kling(2006)]{kling2006incarceration}
J.~R. Kling.
\newblock Incarceration length, employment, and earnings.
\newblock \emph{American Economic Review}, 96\penalty0 (3):\penalty0 863--876,
  2006.

\bibitem[Koles{\'a}r et~al.(2015)Koles{\'a}r, Chetty, Friedman, Glaeser, and
  Imbens]{kolesar2015identification}
M.~Koles{\'a}r, R.~Chetty, J.~Friedman, E.~Glaeser, and G.~W. Imbens.
\newblock Identification and inference with many invalid instruments.
\newblock \emph{Journal of Business \& Economic Statistics}, 33\penalty0
  (4):\penalty0 474--484, 2015.

\bibitem[Kourentzes et~al.(2014)Kourentzes, Barrow, and
  Crone]{kourentzes_neural_2014}
N.~Kourentzes, D.~K. Barrow, and S.~F. Crone.
\newblock Neural network ensemble operators for time series forecasting.
\newblock \emph{Expert Systems with Applications}, 41\penalty0 (9):\penalty0
  4235--4244, July 2014.

\bibitem[Lewis and Syrgkanis(2018)]{lewis2018adversarial}
G.~Lewis and V.~Syrgkanis.
\newblock Adversarial generalized method of moments.
\newblock \emph{arXiv preprint arXiv:1803.07164}, 2018.

\bibitem[Mueller-Smith(2015)]{mueller2015criminal}
M.~Mueller-Smith.
\newblock The criminal and labor market impacts of incarceration.
\newblock \emph{Unpublished Working Paper}, 18, 2015.
\newblock URL
  \url{https://sites.lsa.umich.edu/mgms/wp-content/uploads/sites/283/2015/09/incar.pdf}.

\bibitem[Newey and Powell(2003)]{newey2003instrumental}
W.~K. Newey and J.~L. Powell.
\newblock Instrumental variable estimation of nonparametric models.
\newblock \emph{Econometrica}, 71\penalty0 (5):\penalty0 1565--1578, 2003.

\bibitem[Parzen(1962)]{parzen_estimation_1962}
E.~Parzen.
\newblock On {Estimation} of a {Probability} {Density} {Function} and {Mode}.
\newblock \emph{The Annals of Mathematical Statistics}, 33\penalty0
  (3):\penalty0 1065--1076, Sept. 1962.

\bibitem[Paszke et~al.(2019)Paszke, Gross, Massa, Lerer, Bradbury, Chanan,
  Killeen, Lin, Gimelshein, Antiga, Desmaison, Kopf, Yang, DeVito, Raison,
  Tejani, Chilamkurthy, Steiner, Fang, Bai, and Chintala]{NEURIPS2019_9015}
A.~Paszke, S.~Gross, F.~Massa, A.~Lerer, J.~Bradbury, G.~Chanan, T.~Killeen,
  Z.~Lin, N.~Gimelshein, L.~Antiga, A.~Desmaison, A.~Kopf, E.~Yang, Z.~DeVito,
  M.~Raison, A.~Tejani, S.~Chilamkurthy, B.~Steiner, L.~Fang, J.~Bai, and
  S.~Chintala.
\newblock Pytorch: An imperative style, high-performance deep learning library.
\newblock In H.~Wallach, H.~Larochelle, A.~Beygelzimer, F.~d~Alch\'{e}-Buc,
  E.~Fox, and R.~Garnett, editors, \emph{Advances in Neural Information
  Processing Systems 32}, pages 8024--8035. Curran Associates, Inc., 2019.

\bibitem[Puli and Ranganath(2019)]{Puli2019}
A.~M. Puli and R.~Ranganath.
\newblock Generalized control functions via variational decoupling.
\newblock \emph{CoRR}, abs/1907.03451, 2019.
\newblock URL \url{http://arxiv.org/abs/1907.03451}.

\bibitem[Singh et~al.(2019)Singh, Sahani, and Gretton]{Singh2019}
R.~Singh, M.~Sahani, and A.~Gretton.
\newblock Kernel instrumental variable regression.
\newblock \emph{arXiv preprint arXiv:1906.00232}, 2019.
\newblock URL \url{http://arxiv.org/abs/1906.00232}.

\bibitem[Stock and Watson(2004)]{stock_combination_2004}
J.~H. Stock and M.~W. Watson.
\newblock Combination forecasts of output growth in a seven-country data set.
\newblock \emph{Journal of Forecasting}, 23\penalty0 (6):\penalty0 405--430,
  Sept. 2004.

\bibitem[Swanson and Hern{\'a}n(2018)]{swanson2018challenging}
S.~A. Swanson and M.~A. Hern{\'a}n.
\newblock The challenging interpretation of instrumental variable estimates
  under monotonicity.
\newblock \emph{International journal of epidemiology}, 47\penalty0
  (4):\penalty0 1289--1297, 2018.

\bibitem[Venter(1967)]{venter_estimation_1967}
J.~H. Venter.
\newblock On {Estimation} of the {Mode}.
\newblock \emph{The Annals of Mathematical Statistics}, 38\penalty0
  (5):\penalty0 1446--1455, Oct. 1967.

\bibitem[Wang and Blei(2019)]{wang2019blessings}
Y.~Wang and D.~M. Blei.
\newblock The blessings of multiple causes.
\newblock \emph{Journal of the American Statistical Association}, \penalty0
  (just-accepted):\penalty0 1--71, 2019.

\end{thebibliography}
\bibliographystyle{abbrvnat}

\appendix
\section{Appendix}
\subsection{Proof of Theorem 1}
The next theorem formalizes this intuition and shows, in particular, that MODE-IV asymptotically identifies
and consistently estimates the causal effect.
\begin{theorem}
  Fix a test point $(t,x)$ and let $\wald{1}, \dots, \wald{k}$ be estimators of the causal effect of $t$ at $x$ corresponding to $k$ (possibly invalid) instruments. 
  E.g., $\wald{j} = \hat{f}_j(t,x)$.
  Denote the true effect as $\beta = E[y|\text{do}(t), x]$.
  Suppose that
  \begin{enumerate}
  \item  (consistent estimators) $\wald{j} \to \beta_j$ almost surely for each instrument. In particular, $\beta_j = \beta$ whenever the $j$th instrument is valid.
  \item (valid mode) At least $p\%$ of the instruments are valid, and no more than $p\% - \epsilon$ of the invalid instruments agree on an effect. That is, $p\%$ of the instruments yield the same estimand if and only if all of those instruments are valid. \vv{awkard sentence}
  \end{enumerate}
  Let $[\lb, \ub]$ be the smallest interval containing $p\%$ of the instruments and let $\modalint = \{i : \lb < \wald{i} < \ub\}$.
  Then,
  \[
    \sum_{i \in \modalint} \hat{w}_i \wald{i} \to \beta
\]
almost surely, where $\hat{w}_i,w_i$ are any non-negative set of weights such that each $\hat{w}_i \to w_i$ a.s. and $\sum_{i \in \modalint} w_i = 1$.
Further suppose that the individual estimators are also asymptotically normal,
\[
  \sqrt{n}\wald{j} \to N(\beta_j,\sigma_j^2)
\]
for each instrument. Then it also holds that the modal estimator is asymptotically normal:
  \[
   \sqrt{n} \sum_{i \in \modalint} \hat{w}_i \wald{i} \to N(\beta, \sum_iw_i^2\sigma_i^2).
 \]
\end{theorem}
\begin{proof}
  First we argue that $\modalint$ converges to a set that contains only valid instruments. All valid instruments converge to a common value $\beta$. The distance between any two valid instruments is at most twice the distance between $\beta$ and the furthest valid instrument.
  Since at least $p\%$ of the instruments are valid, this means that there is an interval (containing the mode) with distance going to $0$ that contains $p\%$ of the instruments. Eventually this must be the smallest interval containing $p\%$ of the instruments, because the limting $\beta_j$ of the invalid instruments are spaced out by assumption.

  The result follows by continuous mapping. \vv{is there anything interesting to say about discontinuity points?}
\end{proof}

\subsection{Relaxing independence of instrumental variables}
\label{sec:relax}

\begin{figure}[t]
  \begin{center}
  \includegraphics[width=0.49\textwidth]{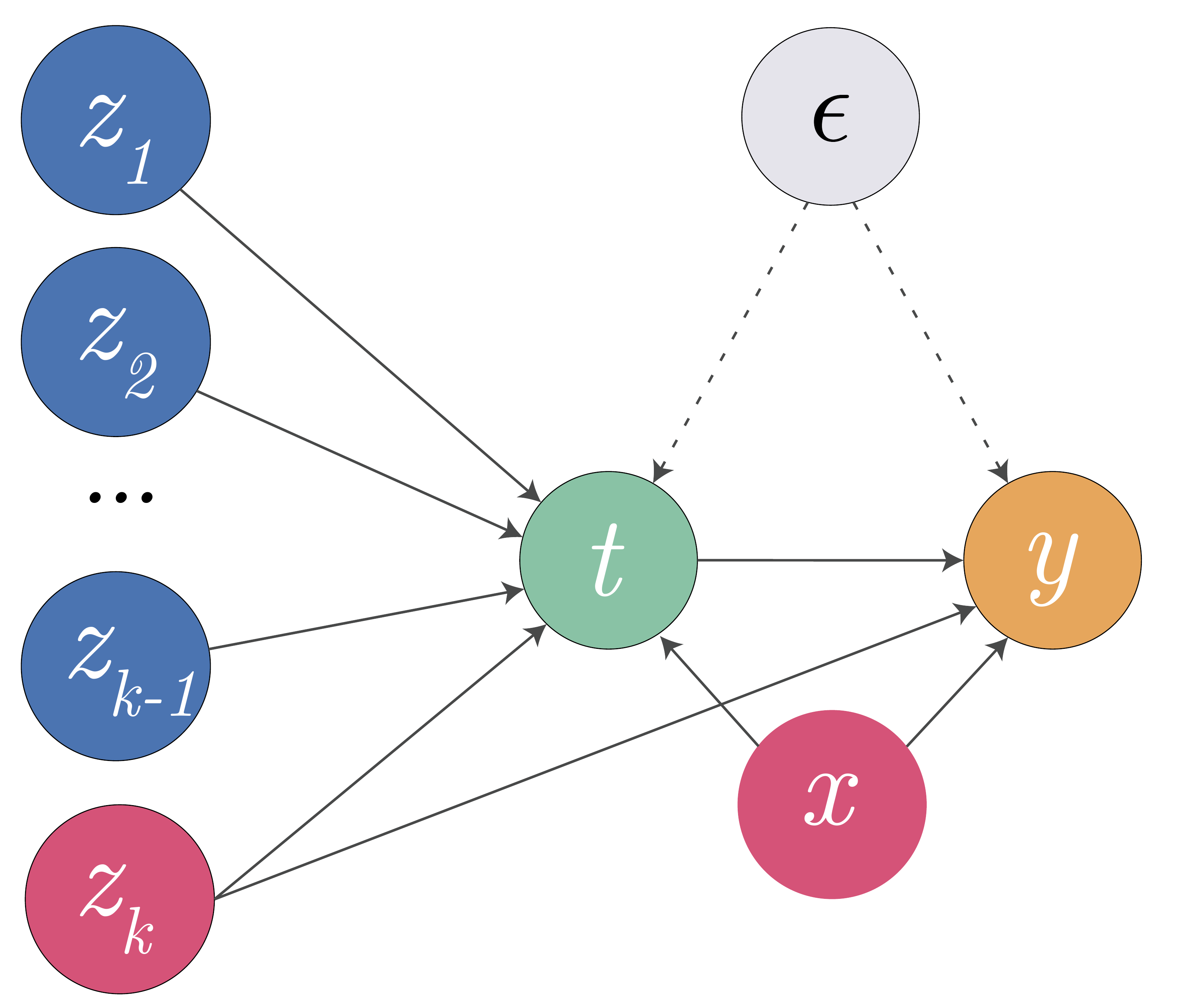}
  \includegraphics[width=0.49\textwidth]{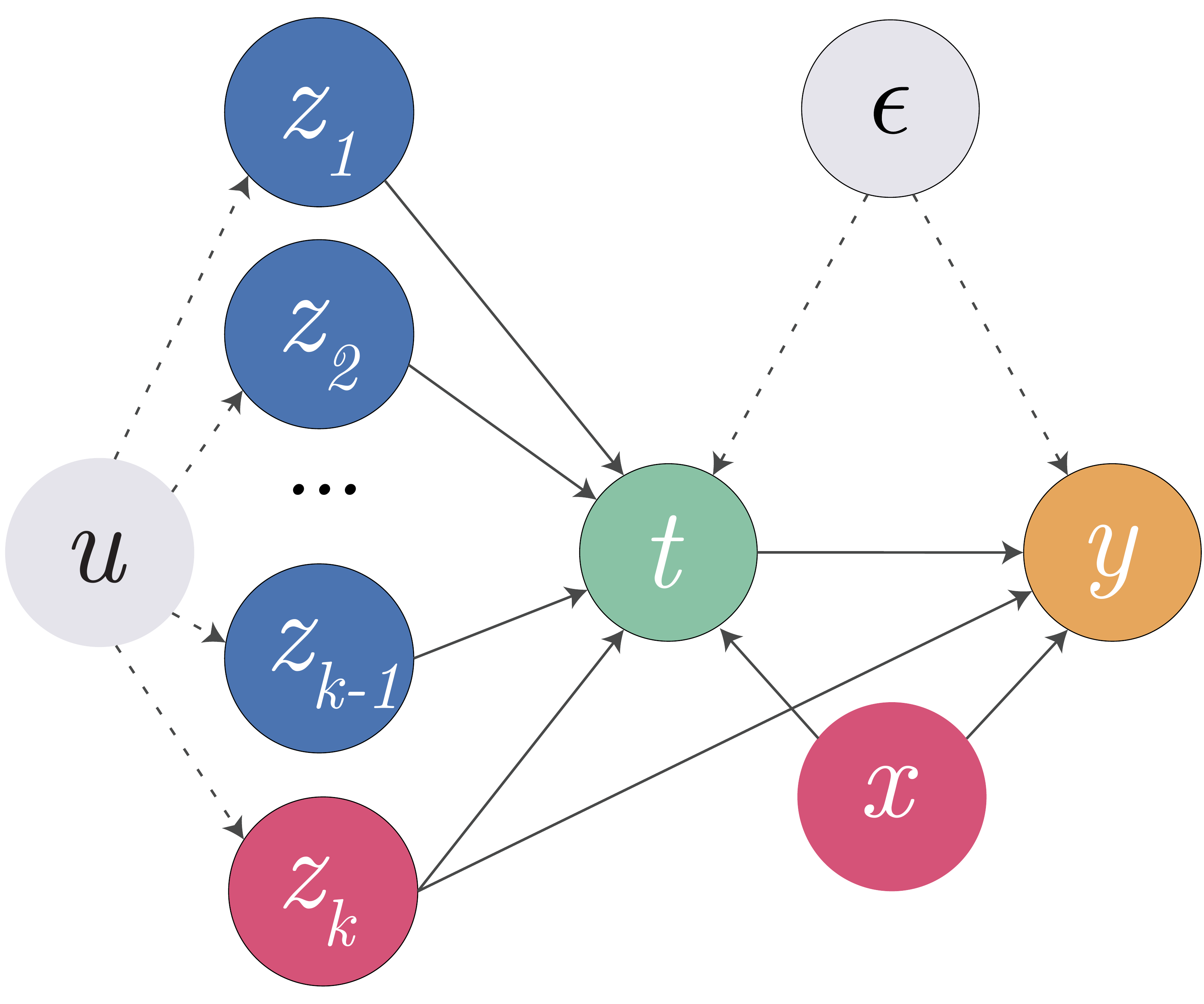}
  \caption{We primarily focus on the setting on the left where each candidate, $z_i$, is independent. I is also possible to apply the method in the setting on the right where the candidates share a common cause, more care is needed. See the discussion in section \ref{sec:relax}.}
  \label{fig:iv}
  \end{center}
\end{figure}

For simplicity, we presented ModeIV in the context of independent candidate instruments. This setting is shown in Figure \ref{fig:iv} (\emph{left}) where we have $k$ candidates, $\{z_i : i\in 1, \dots, k \}$, some of which are valid (shown in blue), and some of which are invalid (e.g. $z_k$ shown in pink has a direct effect on the response). This independent candidates setting is most common where the instruments are explicitly randomized: e.g. in judge fixed effects where the selection of judges is random. 

A more complex setting is shown in Figure \ref{fig:iv} (\emph{right}). Here, the candidates share a common cause, $u$. In this scenario, if $u$ is not observed, each of the previously valid instruments (e.g. $z_1$, $z_2$ and $z_{k-1}$ in the figure) are no longer valid because they fail the unconfounded instrument assumption via the backdoor path  $z_1 \leftarrow u \rightarrow z_k \rightarrow y$.
However, if we condition on all the candidates that have a direct effect on $y$ and treat them as observed confounders, we block this path which allows for valid inference. Of course we do not know which of the candidates have a direct effect, so when building an ensemble, for each candidate $z_i$, we treat all $z_{j \neq i}$ as observed confounds to block these potential backdoor paths. This addresses the issue as long as there is not some $z_{k+1}$ which is not part of our candidate set, but nevertheless opens up a backdoor path $z_1 \leftarrow u \rightarrow z_{k+1} \rightarrow y$.

If $u$ is observed, we can simply control for it. This suggests a natural alternative approach would be to try to estimate $u$ and control for its effect, using an approach analogous to \citet{wang2019blessings}. We plan to investigate this in future work. 

\subsection{Additional experimental details}
\label{sec:add_experiments}

\begin{table}
\centering
\resizebox{\columnwidth}{!}{%
\begin{tabular}{lllllll}
\toprule
{} &                           50  &                           60  &                           70  &                           80  &                           90  &                           100 \\
\midrule
DeepIV (valid) &   0.036 \tiny{$\pm$ (0.0048)} &    0.042 \tiny{$\pm$ (0.003)} &  0.0371 \tiny{$\pm$ (0.0032)} &   0.0364 \tiny{$\pm$ (0.003)} &  0.0326 \tiny{$\pm$ (0.0021)} &  0.0278 \tiny{$\pm$ (0.0021)} \\
MODE-IV 20     &  0.0525 \tiny{$\pm$ (0.0062)} &  0.0483 \tiny{$\pm$ (0.0049)} &  0.0478 \tiny{$\pm$ (0.0049)} &  0.0473 \tiny{$\pm$ (0.0048)} &  0.0466 \tiny{$\pm$ (0.0047)} &    0.04 \tiny{$\pm$ (0.0038)} \\
MODE-IV 30     &  0.0525 \tiny{$\pm$ (0.0062)} &  0.0483 \tiny{$\pm$ (0.0049)} &  0.0478 \tiny{$\pm$ (0.0049)} &  0.0473 \tiny{$\pm$ (0.0048)} &  0.0467 \tiny{$\pm$ (0.0047)} &  0.0399 \tiny{$\pm$ (0.0038)} \\
MODE-IV 40     &  0.0524 \tiny{$\pm$ (0.0062)} &  0.0483 \tiny{$\pm$ (0.0049)} &  0.0478 \tiny{$\pm$ (0.0049)} &  0.0473 \tiny{$\pm$ (0.0048)} &  0.0468 \tiny{$\pm$ (0.0047)} &  0.0399 \tiny{$\pm$ (0.0038)} \\
MODE-IV 50     &  0.0525 \tiny{$\pm$ (0.0062)} &  0.0483 \tiny{$\pm$ (0.0049)} &  0.0479 \tiny{$\pm$ (0.0049)} &  0.0474 \tiny{$\pm$ (0.0048)} &  0.0466 \tiny{$\pm$ (0.0047)} &  0.0398 \tiny{$\pm$ (0.0038)} \\
Mean           &  0.0529 \tiny{$\pm$ (0.0059)} &   0.0498 \tiny{$\pm$ (0.005)} &  0.0498 \tiny{$\pm$ (0.0052)} &  0.0461 \tiny{$\pm$ (0.0047)} &  0.0484 \tiny{$\pm$ (0.0048)} &   0.0403 \tiny{$\pm$ (0.004)} \\
Deepiv (all)   &   0.1637 \tiny{$\pm$ (0.011)} &  0.1744 \tiny{$\pm$ (0.0075)} &  0.2078 \tiny{$\pm$ (0.0069)} &   0.185 \tiny{$\pm$ (0.0087)} &  0.1387 \tiny{$\pm$ (0.0081)} &  0.0297 \tiny{$\pm$ (0.0018)} \\
\bottomrule
\end{tabular}
}
\vspace{.4em}
\caption{Average absolute bias in estimation of the conditional average treatment effect. The ensemble methods tended to have slightly larger bias than the optimal model, but far less than the naive approach which uses all instruments. The mean aggregation function performs relatively well on this task, but this approach comes with no guarantees, so it degrades in settings with more bias.}
\label{tab:grad_bias}
\end{table}

\paragraph{Network architectures and experimental setup}
All experiments used the same neural network architectures to build up hidden representations for both the treatment and response networks used in DeepIV, and differed only in their final layers. Given the number of experiments that needed to be run, hyper-parameter tuning would have been too expensive, so the hyper-parameters were simply those used in the original DeepIV paper. In particular, we used three hidden layers with 128, 64, and 32 units respectively and ReLU activation functions. The treatment networks all used mixture density networks with 10 mixture components and the response networks were trained using the two sample unbiased gradient loss \citep[see][equation 10]{hartford2017deep}. We used our own PyTorch \citep{NEURIPS2019_9015} re-implementation of DeepIV to run the experiments.

For all experiments, $10\%$ of the original training set was kept aside as a validation set. The demand simulations had 90 000 training examples, 10 000 validation examples and 50 000 test examples. The Mendelian randomization simulations had 360 000 training examples, 40 000 validation examples and 50 000 test examples. All mean squared error numbers reported in the paper are calculated with respect to the true $y$ (with no confounding noise added) on a uniform grid of 50 000 treatment points between the $2.5^{th}$ and $97.5^{th}$ percentiles of the training distribution. 95\% confidence intervals around mean performance were computed using Student's $t$ distribution.




The networks were trained on a large shared compute cluster that has around 100 000 CPU cores. Because each individual network was relatively quick to train (less than 10 minutes on a CPU), we used CPUs to train the networks. This allowed us to fit the large number of networks needed for the experiments. Each experiment was run across 30 different random seeds, each of which required 10 (demand simulation) and 102 (Mendialian randomization) network fits. In total, across all experimental setups, random seeds, ensembles, etc. approximately 100 000 networks were fit to run all the experiments.

\paragraph{Biased demand simulation}
The biased demand simulation code was modified from the \href{https://github.com/jhartford/DeepIV/blob/master/experiments/data_generator.py}{public DeepIV implementation}. In the original DeepIV implementation, both the treatment and response are transformed to make them approximately mean zero and standard deviation 1; we left these constants unchanged ($p_{\text{std}} = 3.7, p_{\mu} = 17.779, y_{\text{std}} = 158, y_\mu = -292.1$). The observed features include a time feature, $t\sim \text{unif}(0, 10)$, and $x \sim \text{Categorical}(\frac{1}{7}, \dots,\frac{1}{7})$, a one-hot encoding of 7 different equally likely ``customer types''; each type modifies the treatment via the coefficient $\beta^{(x)}=[1, 2, \dots, 7]^T$. These values are unchanged from the original \citeauthor{hartford2017deep} data generating process. 
We introduce multiple instruments, $z_{1:k}$, whose effect on the treatment, $p$, and response, $y$, is via two different linear maps $\beta^{(zp)}\in \Re^k$ and  $\beta^{(zy)}\in \Re^k$; each of the coefficients in these vectors are sampled independently so $\beta^{(z*)}_i\sim\text{unif}(0.5, 1.5)$, with the exception of the valid instruments where $\beta^{(zy)}_i = 0 $ for all $i\in \valid$. The $\gamma$ parameter scales the amount of bias introduced via exclusion violations; note that because the $\sin$ varies between $[-1,1]$, we scale up this bias by a factor of $60$ so that the effect it introduces is of the same order of magnitude as the variation in the original \citeauthor{hartford2017deep} data generating process (where $\text{std. dev}(y)\approx y_{\text{std}} = 158$).
The full data generating process is as follows,
\begin{align*}
    \textcolor{Maroon}{z_{1:k}}, \nu &\sim \mathcal{N}(0, 1) \quad t\sim \text{unif}(0, 10) \quad e\sim\mathcal{N}(\rho \nu, 1 - \rho^2)\\
    p' &= 25 + (\textcolor{Maroon}{z^T \beta^{(zp)}} + 3)\psi(t) + \nu \\
    y' &= 100 + 10x^T \beta^{(x)} \psi(t) + \underbrace{(x^T \beta^{(x)} \psi(t) - 2)}_{\text{Treatment effect}}p' + \underbrace{\textcolor{Maroon}{\gamma 60 \sin(z^T \beta^{(zy)})}}_{\text{Exclusion violation}}  + e \\
    p &= (p' - p_{\text{std}}) / p_{\mu} \qquad y = (y' - y_{\text{std}}) / y_{\mu}
\end{align*}

\paragraph{Mendelian randomization simulation}

This data generating process closely follows that of Simulation 1 of \citet{hartwig2017robust}, but was modified to include heterogeneous treatment effects. This description is an abridged version of that given in \citeauthor{hartwig2017robust}; we refer the reader to \citet{hartwig2017robust} for more detail on the choice of parameters, etc.. The instruments are the genetic variables, $z_i$, which were generated by sampling from a Binomial (2, $p_i$) distribution, with $p_i$ drawn from a Uniform(0.1 ,0.9) distribution, to mimic bi-allelic SNPs in Hardy-Weinberg equilibrium. The parameters that modulate the genetic variable effect on the treatment are given by, $\alpha_i = \frac{\sqrt{0.1}}{\sigma_{zx}}\nu_i$, where $\nu_i\sim \text{unif}(0.01, 0.2)$ and $\sigma_{zx} = \text{std. dev}(\sqrt{0.1} \sum_i \nu_i z_i)$. 
Similarly, the exclusion violation parameters, $\delta_i=\frac{|\invalid|\sqrt{0.1}}{k\sigma_{zy}}\nu_i$, where again $\nu_i\sim \text{unif}(0.01, 0.2)$ and $\sigma_{zy} = \text{std. dev}(\sqrt{0.1} \sum_i \nu_i z_i)$. Note that $\delta_i$ is scaled by $\frac{|\invalid|}{k}$ (the proportion of invalid instruments), which ensures that the average amount of bias introduced is constant as the number of invalid instruments vary.
Error terms $u$, $\epsilon_{x}$, $\epsilon_{y}$ were independently generated from a normal distribution, with mean 0 and variances $\sigma^2_u,\sigma^2_y$ and $\sigma^2_y$, respectively, whose values were chosen to set the variances of $u$, $x$ and $y$ to one. 
These scaling parameters are chosen to enable an easy interpretation of the average treatment effect, $\beta$: with this scaling, $\beta = 0.1$ implies that one standard deviation of $t$ causes a $0.1$ standard deviation of $y$ and hence the causal effect of $t$ on $y$ explains $0.1^2 = 0.01=1\%$ of the variance of $y$. 
The only place our simulation differs from \citeauthor{hartwig2017robust}, is the treatment effect is a function of observed coefficients, $x\in \Re^{10}$, with each $x_i\sim \text{uniform}(-0.5,0.5)$, and the treatment effect is defined as, $\beta(x) := \text{round}(x^T\gamma^{(xt)}, 0.1)$, with $\gamma^{(xt)}$ a sparse vector of length 10, with three non-zeros $\gamma^{(xt)}_i\sim \text{uniform}(0.2,0.5)$. The resulting true $\beta(x)$, takes on values in $\{-0.3, -0.2, \dots, 0.2, 0.3\}$. We also use 100 genetic variants instead of the 30 used in \citeauthor{hartwig2017robust}, so that we could test ModeIV in a larger scale scenario. The resulting data generating process is given by,

\begin{align*}
    &z_i \sim \text{Binomial}(2, p_i) \quad \text{for $i$ in }[1\dots K],\quad \beta(x) := \text{round}(x^T\gamma^{(xt)}, 0.1).\\
    &t := \sum_{j=1}^K \alpha_j z_j + \rho u + \epsilon_x \quad y :=\beta(x) t + \sum_{j=1}^K \delta_j z_j + u + \epsilon_y  \nonumber
\end{align*}

\end{document}